\documentclass[a4paper,UKenglish,cleveref, autoref, thm-restate]{lipics-v2021}
\usepackage[textsize=tiny,disable]{todonotes}
\usepackage{cite}
\usepackage{mathtools}
\usepackage{subcaption}

\pdfoutput=1 
\hideLIPIcs  

\newcommand{\crn}[1]{\mathcal{#1}}
\newcommand{\species}{\mathcal{S}}
\newcommand{\reactions}{\mathcal{R}}
\newcommand{\size}[1]{\lvert #1 \rvert}
\newcommand{\N}{\mathbb{N}}

\newcommand{\U}{\mathcal{U}}
\newcommand{\Oh}{\mathcal{O}}
\newcommand{\RN}[1]{%
  \textup{\uppercase\expandafter{\romannumeral#1}}%
}
\renewcommand\vec{\mathbf}

\def\stepsto{\mathop{\rightarrowtail}\limits}
\def\reaches{\mathop{\leadsto}\limits}

\def\rxn{\mathop{\rightarrow}\limits}

\def\revrxn{\mathop{\rightleftarrows}\limits}
\def\metarevrxn{\mathop{\rightleftharpoons}\limits}

\newcommand{\Kcrn}{\mathrm{K}_\mathrm{crn}}
\newcommand{\Kc}{\mathrm{\widetilde{K}s}}
\newcommand{\Ks}{\mathrm{Ks}}   
\newcommand{\K}{\mathrm{K}}

\nolinenumbers

\title{Optimal Information Encoding in Chemical Reaction Networks} 

\author{Austin {Luchsinger}}{Electrical and Computer Engineering, University of Texas at Austin, TX, USA \and \url{https://sites.google.com/utexas.edu/austinluchsinger/home}}{amluchsinger@utexas.edu}{https://orcid.org/0000-0002-8180-9762}{Carroll H.\ Dunn Endowed Graduate Fellowship in Engineering}

\author{David {Doty}}{Computer Science, University of California--Davis, CA, USA \and \url{https://web.cs.ucdavis.edu/~doty/}}{doty@ucdavis.edu}{https://orcid.org/0000-0002-3922-172X}{NSF grants 2211793, 1900931, and CAREER-1844976}

\author{David {Soloveichik}}{Electrical and Computer Engineering, University of Texas at Austin, TX, USA \and \url{https://users.ece.utexas.edu/~soloveichik/}}{david.soloveichik@utexas.edu}{https://orcid.org/0000-0002-2585-4120}{NSF grant 1901025, Sloan Foundation Research Fellowship}

\authorrunning{A.\ Luchsinger, D.\ Doty, and D.\ Soloveichik}

\Copyright{Austin Luchsinger, David Doty, and David Soloveichik} 



\category{} 

\relatedversion{} 

\acknowledgements{The authors thank Eric Severson for the invaluable discussions on chemical reaction network computation that helped lead to the results presented here.}

\EventEditors{John Q. Open and Joan R. Access}
\EventNoEds{2}
\EventLongTitle{42nd Conference on Very Important Topics (CVIT 2016)}
\EventShortTitle{CVIT 2016}
\EventAcronym{CVIT}
\EventYear{2016}
\EventDate{December 24--27, 2016}
\EventLocation{Little Whinging, United Kingdom}
\EventLogo{}
\SeriesVolume{42}
\ArticleNo{23}

\begin{document}
\maketitle

\begin{abstract}
    Discrete chemical reaction networks formalize the interactions of molecular species in a well-mixed solution as stochastic events. Given their basic mathematical and physical role,
    the computational power of chemical reaction networks has been widely studied in the molecular programming and distributed computing communities.
    While for Turing-universal systems there is a universal measure of optimal information encoding based on Kolmogorov complexity, 
    chemical reaction networks are not Turing universal unless error and unbounded molecular counts are permitted.
    Nonetheless, here we show that the optimal number of reactions to generate a specific count $x \in \mathbb{N}$ with probability $1$ is asymptotically equal to a ``space-aware'' version of the Kolmogorov complexity of $x$,
    defined as $\Kc(x) = \min_p\left\{\size{p} / \log \size{p} + \log(\texttt{space}(\U(p))) : \U(p) = x \right\}$, where $p$ is a program for universal Turing machine $\U$.
    This version of Kolmogorov complexity incorporates not just the length of the shortest program for generating $x$, but also the space usage of that program.
    Probability $1$ computation is captured by the standard notion of stable computation from distributed computing, but we limit our consideration to chemical reaction networks 
    obeying a stronger constraint:
    they ``know when they are done'' in the sense that they produce a special species to indicate completion.  
    As part of our results, we develop a module for encoding and unpacking any $b$ bits of information via $O(b/\log{b})$ reactions, which is information-theoretically optimal for incompressible information.
    Our work provides one answer to the question of how succinctly chemical self-organization can be encoded---in the sense of generating precise molecular counts of species as the desired state.
\end{abstract}

\clearpage
\pagenumbering{arabic}

\section{Introduction}\label{sec:introduction}

In potential biochemical, nanotechnological, or medical applications, synthetic chemical computation could allow for the re-programming of biological regulatory networks and the insertion of control modules where traditional electronic controllers are not feasible.
Understanding the design principles of chemical information processing also may achieve better understanding of the complex information processing that occurs in biological chemical interactions.

Discrete chemical reaction networks, also called stochastic chemical reaction networks, is a formal model of chemical kinetics in a well-mixed solution.
While in continuous chemical kinetics, continuous concentrations change in time governed by ordinary differential questions, 
here the state consists of non-negative integer molecular counts of the species, and reaction events occur stochastically as a continuous time Markov process. 
Closely related models include population protocols in distributed computing~\cite{angluin2004computation}, as well as models without stochastic kinetics such as Petri nets~\cite{petrinets}, vector addition systems~\cite{karp1969parallel} and commutative semigroups~\cite{cardoza1976exponential}.
The model is particularly relevant when some species are present in small molecular counts,
which are not well-approximated by continuous concentrations~\cite{gillespie2007stochastic}; this regime is germane for small volumes such as that of a cell, natural or artificial.
For the rest of this paper, the acronym CRNs (Chemical Reaction Networks) refers to  the discrete model.

Typically the ensuing sequence of reactions can be predicted only stochastically since multiple reactions compete with each other. 
Nonetheless certain behaviors are independent of the order in which reactions happen to occur.
Such probability $1$ behavior is formalized using the notion of stable computation.
For example the reactions $X_1 \to 2Y$ and $X_2 + Y \to \emptyset$ compute the function $f(x_1,x_2) = \max(2x_1 - x_2, 0)$ regardless of the order in which reactions happen.
Below when we say that a CRN computes something, we mean it in the sense of stable computation.
It is known that stably computing CRNs are not Turing-universal~\cite{soloveichik2008computation}, but instead are limited to computing semilinear predicates and functions~\cite{angluin2006stably,chen2014deterministic}. 
However, the scaling of the computational power of CRNs with the number of reactions and species still lacks a tight and general characterization. 

Prior approaches to answering the question of reaction or species complexity---in the equivalent language of population protocols---have focused largely on predicate computation and can be divided into two groups.
(We should point out that the literature makes the important distinction between population protocols with and without a ``leader,'' which is equivalent to starting with a single copy of a distinguished species in the initial state.
The prior results described here as well as our work correspond to protocols \emph{with} a leader.)
The first line of work focuses on specific predicates---with the prototypical choice being the so-called ``counting predicates'' in which the task is to decide whether the count of the input species is at least some threshold $x \in \mathbb{N}$~\cite{blondin2018large,czerner2021lower,leroux2022state}.
In particular, close upper and lower bounds were developed:
for infinitely many $x$, the predicate can be stably decided with $\Oh(\log{\log{x}})$ species~\cite{blondin2018large}, and $\Oh((\log{\log{x}})^{1/2 - \epsilon})$ species are required~\cite{leroux2022state}.
\todo{[For journal version] DS: somewhere point out when these bounds match ours}

Other work has focused on the more general characterization of predicate computation.
It is well-known that semilinear predicates can be characterized in terms of Presburger arithmetic, the first-order theory of addition.
It was subsequently shown that a CRN can decide a semilinear predicate with the number of species scaling polynomially with the size of the corresponding Presburger formula~\cite{blondin2019succinct,czerner2022fast}.
There are also provable tradeoffs between the speed of computation and the number of species (e.g.,~\cite{alistarh2018space,doty2022time,berenbrink2020optimal}).
We do not consider the time-complexity of CRNs further in this paper.

While the prior work described above involves stably deciding a counting predicate where the system recognizes if the count of some species is at least $x$, we investigate the problem of generating exactly $x$ copies of a particular species $Y$, starting from a single copy of another species $L$.
This idea of generation is natural for engineers of these systems who may wish to prepare a particular configuration to be used in a downstream process, and captures a certain form of chemical self-organization.
(We note the conceptual connection to another type of self-organization: leader-election, in which we want to end up with exactly one molecule of a species, starting from many~\cite{berenbrink2020optimal}.)
Our constructions can be adopted to deciding the counting predicates with only a constant more reactions---giving a novel upper bound on the number of reactions (see Open Questions).
It is also worth noting that other complexity questions have been investigated for CRNs, such as ``the size of the smallest chemical reaction network that approximates a desired distribution'' \cite{cappelletti2020stochastic}.

The goal of this paper is to connect the complexity of the most compact CRN for generating $x$ to the well-known measures of the optimal ``description length'' of $x$.
Kolmogorov complexity, a widely recognized concept across various disciplines in computer science and information theory, serves as a universal, broadly accepted measure of description length~\cite{li2008introduction}. 
This notion quantifies the complexity of an object, such as a string or a number, by the length of the shortest program that produces it. 
While the minimal number of species or reactions to generate count $x$ cannot be connected to the canonical Kolmogorov complexity, we provide tight asymptotic bounds to a modification of Kolmogorov complexity $\Kc$ (\Cref{eq:Kc}).
As this quantity incorporates not only the length of the shortest program to produce $x$, but also the space (memory) usage of the program, it can be called ``space-aware.''
Unlike the canonical Kolmogorov complexity, $\Kc$ is computable.

Our quantity $\Kc$ characterizes the CRN complexity of generating $x$ in the range from $\Oh(\log \log x)$ for highly ``compressible'' $x$ to $\Oh(\log x/\log \log x)$ for ``incompressible'' $x$.
The module we develop for optimally encoding $b$ bits of information with $\Oh(b/\log b)$ reactions via a permutation code may be of independent interest.
The encoded information could be used for other purposes than for generating a desired amount of some species, which justifies a more general interpretation of our work as studying the encoding information in CRNs.

\section{Preliminaries}
\label{sec:preliminaries}
We use notation from \cite{chalk2019composable, severson2019composable} and stable computation definitions from \cite{angluin2006stably, chugg2018output} for (discrete) chemical reaction networks.
Let $\N$ denote the nonnegative integers.
For any finite set $\species$ (of species), we write $\N^{\species}$ to mean the set of functions $f: \species \rightarrow \N$. 
Equivalently, $\N^{\species}$ can be interpreted as the set of vectors indexed by the elements of $\species$,
and so $\vec{c} \in \N^{\species}$ specifies nonnegative integer counts for all elements of $\species$.
For $\vec{a},\vec{b} \in \N^\species$, we write $\vec{a} \leq \vec{b}$ if $\vec{a}(i) \leq \vec{b}(i), \forall i$.

\subsection{Chemical Reaction Networks}
A \emph{chemical reaction network} (CRN) $\crn{C} = (\species, \reactions)$ is defined by a finite set $\species$ of species, and a finite set $\reactions$ of reactions where each reaction is a pair $\langle \vec{r},\vec{p}\rangle \in \N^\species \times \N^\species$ that denotes the \textit{reactant} species consumed by the reaction and the \textit{product} species generated by the reaction.
For example, given $\species = \{A,B,C\}$, the reaction $\langle (2,0,0), (0,1,1) \rangle$ represents $2A \rxn B + C$. 
Although the definition allows for more general stoichiometry, in this paper we only consider third-order reactions (with at most three reactants and three products). 
For reversible reactions, we will use the notation $A + B \revrxn C + D$ to mean $A + B \rxn C + D$ and $C + D \rxn A + B$.
We say that the \textit{size} of a CRN (denoted $\size{\crn{C}}$) is simply the number of reactions in $\reactions$.\footnote{When considering systems with third-order reactions it is clear that $\size{\reactions}^{1/6} \leq \size{\species} \leq 6\size{\reactions}$.} 

A \textit{configuration} $\vec{c}\in \N^\species$ of a CRN assigns integer counts to every species $s \in \species$.
When convenient, we use the notation $\{n_1S_1,n_2S_2,\dots,n_kS_k\}$ to describe a configuration with $n_i\in\N$ copies of species $\species_i, \forall i \in[1,k]$. When using this notation, any species $S_j\in\species$ that is not listed is assumed to have a zero count (e.g., given $\species = \{A,B,C\}$, the configuration $\{3A, 2B\}$ has three copies of species $A$, two copies of species $B$, and zero of species $C$).
For two configurations $\vec{a},\vec{b} \in \N^\species$, we say $\vec{b}$ \textit{covers} $\vec{a}$ if $\vec{a} \leq \vec{b}$; in other words, for all species, $\vec{b}$ has at least as many copies as $\vec{a}$.

A reaction $\langle\vec{r},\vec{p}\rangle$ is said to be \textit{applicable} in configuration $\vec{c}$ if $\vec{r} \leq \vec{c}$. 
If the reaction $\langle\vec{r},\vec{p}\rangle$ is applicable, it results in configuration $\vec{c}' = \vec{c} - \vec{r} + \vec{p}$ if it occurs, and we write $\vec{c} \stepsto \vec{c}'$.
If there exists a finite sequence of configurations such that $\vec{c} \stepsto \vec{c}_1 \stepsto \dots \stepsto \vec{c}_n \stepsto \vec{d}$, then we say that $\vec{d}$ is \textit{reachable} from $\vec{c}$ and we write $\vec{c} \reaches \vec{d}$.

In keeping with established definitions for stable computation, we specify an \textit{output species} $Y \in \species$ and a \textit{leader species} $L \in \species$ for stable integer computation.\footnote{For stable \textit{function} computation, an ordered subset of input species $\{X_1,X_2,\dots,X_n\} \subset \species$ is also included; however, stable integer computation would be something along the lines of $f(1) = x$, so a single copy of the leader species serves as the ``input'' here.}
We start from an initial configuration $\vec{i} = \{1 L\}$.
A configuration $\vec{c}$ is \textit{output-stable} if $\forall \vec{d}$ such that $\vec{c} \reaches \vec{d}$, $\vec{c}(Y) = \vec{d}(Y)$.
CRN $\crn{C}$ \textit{stably computes} integer $x$ if, from any configuration $\vec{c}$ that is reachable from input configuration $\vec{i}$, there is an output-stable configuration $\vec{o}$ reachable from $\vec{c}$ with $\vec{o}(Y) = x$.
Note that when considering systems with bounded state spaces like those discussed in this paper, stable computation is equivalent to probability 1 computing.

We also consider a much stronger constraint on CRN computation that specifies a special halting species.
A species $H \in \species$ is a \textit{halting species} if $\forall \vec{c}$ such that $\vec{c}(H) \geq 1$, 
$\vec{c}$ is output stable and $\forall \vec{d}$ where $\vec{c}\reaches\vec{d}$, $\vec{d}(H) \geq 1$. 
We say that a CRN $\crn{C}$ \textit{haltingly computes} an integer $x$ if (1) $\crn{C}$ stably computes $x$ and (2) $\crn{C}$ has a halting species $H$.
Intuitively, a halting CRN knows when it is done---the halting species can initiate some downstream process that is only meant to occur when the computation is finished.

\subsection{Kolmogorov Complexity}
\label{sec:kolmogorov-complexity}
A focus of this paper is the ``optimal description'' of integers.
As such, we often refer to the traditional notion of Kolmogorov complexity which we define here.

Let $\U$ be a universal Turing machine. 
The Kolmogorov complexity for an integer $x$ is the value $\K(x) = \min\{\size{p} : \U(p) = x\}$.
In other words, the Kolmogorov complexity of $x$ is the size of the smallest Turing machine program $p$ that outputs $x$.
This captures the descriptional complexity of $x$ in the sense that a (smaller) description of $x$ can be given to some machine that generates $x$ based on the given description.

We use a ``space-aware'' variant of this quantity which we later connect to the size of the smallest CRN stably computing $x$:
\begin{equation}
\label{eq:Kc}
\Kc(x) = \min\left\{\frac{\size{p}}{{\log \size{p}}} + \log(\texttt{space}(\U(p))) : \U(p) = x \right\}.
\end{equation}
Note that $\Kc(x)$ does not refer to CRNs in any direct way, so the tight asymptotic connection (\Cref{thm:main}) we establish may be surprising.

$\Kc(x)$ is similar to the Kolmogorov complexity variant defined as $\Ks(x) = \min\{\size{p} + \log(\texttt{space}(\U(p))) : \U(p,i) = x[i]\}$ by Allender, Kouck\`{y}, Ronneburger, and Roy~\cite{allender2011pervasive}
in that it additively mixes program size with the log of the space usage. 
There are two differences:
(1) The program size component of $\Kc$ is $\size{p}/\log{\size{p}}$ rather than $\size{p}$.
The intuition is that a single chemical reaction can encode more than one bit of information; thus, a Turing machine program $p$ can be converted to a ``CRN program'' with a number of reactions that is asymptotically smaller than the number of bits of $p$.
(2) $\Ks(x)$ is defined with respect to programs that,
given index $i$ as input,
output $x[i]$, the $i$'th bit of $x$, while our $\Kc(x)$ is defined with respect to programs that (taking no input) directly output all of $x$.
Thus $\Kc(x) \geq \log{\size{x}}$, since the Turing machine must at least store the output integer, while $\Ks(x)$ may be smaller in principle.
Due to the ability of efficient universal Turing machines to simulate each other efficiently,
$\Kc$ (like $\Ks$) is invariant within multiplicative constants to the choice of universal Turing machine $\U$, as long as $\U$ is space-efficient.
Note that if $\Kc$ were not robust to the choice of $\U$, it could hardly be a universal measure.

It is worth noting that unlike $\K(x)$, $\Kc(x)$ is computable.
To see this, one can enumerate all programs for universal Turing machine $\U$ and run them in order from smallest to largest, stopping on the first machine that outputs $x$.
Since the space usage of $\U(p)$ is included in $\Kc$, we can terminate executions as soon as they start using too much space.
This ensures that no execution will run forever, and so we are guaranteed to find the smallest $p$ that outputs $x$. 
\todo{DS: [For journal version]: Include something about max species count bounded by $\Kc$.}

\subsection{Overview}
Here, we give a high level overview for the constructions and results presented in the subsequent sections of this paper.

Our constructions rely on the ability of CRNs to ``efficiently'' simulate space-bounded Turing machines (in terms of program size and space usage, not time) by ``efficiently'' simulating bounded-count register machines.
\Cref{sec:efficient-bounded-RMs} details how to use a combination of previous results to achieve this.
The first half of the section describes how to construct a CRN to faithfully simulate a bounded-count register machine.
The second half of the section shows how to generate a large register machine bound (${2^2}^{n}$) with very few species/reactions ($n$).
While the latter result is from previous work \cite{cardoza1976exponential}, we translate their construction from a commutative semigroup presentation into a chemical reaction network.

In \Cref{sec:encoding}, we present a method for constructing a CRN $\crn{C}_x$ which (optimally) haltingly computes $n$-bit integer $x$ with $\size{\crn{C}_x} = \Oh(n/\log n)$ by using a permutation code (\Cref{thm:incompressible}).
The idea of the construction is to generate a specified permutation and convert that permutation to a mapped target integer $x$.
This construction relies on the ``efficient'' bounded-count register machine and space-bounded Turing machine simulations.

We then show how to use our permutation construction to achieve an optimal encoding (within global multiplicative constants) for algorithmically compressible integers in \Cref{sec:compression}.
Here, we use our permutation code technique to ``unpack'' a Turing machine program that that outputs $x$, resulting in a CRN that haltingly computes $x$ with $\Oh(\Kc)$ reactions (\Cref{thm:compressible}).
Afterwards, we use a result from K\"{u}nnemann et al.~\cite{kunnemann2023coverability} to show that the size of our constructed CRN is within multiplicative constants of the optimal size of a CRN that stably computes $x$, denoted $\Kcrn(x)$ (\Cref{thm:optimality}).
The results of the paper culminate with us connecting $\Kcrn(x)$ and $\Kc$ in \Cref{thm:main} (our main theorem), which is directly implied by the combination of \Cref{thm:compressible} and \Cref{thm:optimality}.

Lastly, we present some open questions for future work in \Cref{sec:conclusion}.

\section{Efficient Simulation of Bounded Register Machines}\label{sec:efficient-bounded-RMs}

\subsection{Register machines}
\label{sec:register_machines}
A register machine is a finite state machine along with a fixed number of registers, each with non-negative integer counts.
The two fundamental instructions for a register machine are increment $\textit{inc}(r_i,s_j)$ and decrement $\textit{dec}(r_i,s_j,s_k)$.
The first instruction increments register $r_i$ and transitions the machine to state $s_j$.
The second instruction decrements register $r_i$ if it is non-zero and transitions the machine to state $s_j$, otherwise the machine just transitions to state $s_k$.
We also consider the more advanced instruction of $\textit{copy}(r_i,r_j,s_k)$, which adds the value of register $r_i$ to register $r_j$, i.e., it is equivalent to the assignment statement $r_j := r_j + r_i$ (note that the value is preserved in $r_i$).
It is clear that \textit{copy} can be constructed with a constant number of register machine states.
In fact, register machines are known to be Turing-universal with three registers \cite{minsky1967computation}.\footnote{Turing-universality has also been shown for machines with two registers, but only when a nontrivial encoding of the input/output is allowed \cite{minsky1967computation,schroeppel1972two}.}

In \cite{soloveichik2008computation}, a simple CRN construction was shown to simulate register machines with some possibility of error (thus not directly compatible with stable computation).
The source of the error is due to the zero-checking in a \textit{dec} instruction.
For the simulation, the CRN has a finite set of species (one for each register and one for each state of the register machine) and a finite set of reactions (one for each instruction in the register machine program).
Each $\textit{inc}(r_i,s_j)$ instruction corresponds to the reaction $S_{j'} \rxn R_i + S_j$, and each $\textit{dec}(r_i,s_j,s_k)$ instruction to two reactions $S_{j'} + R_i \rxn S_j$ and $S_{j'} \rxn S_k$.
In the chemical reaction network implementation of a \textit{dec} instruction, the two reactions are competing for the state species $S_{j'}$.
While in general this is an unavoidable problem,
in the special case that the maximum value in our counters is bounded by a constant, 
we can remedy this following the idea from~\cite{lipton1976reachability} as follows.

Let's consider bounded registers that can contain a value no greater than $b \in \N$.
For each register $r_i$, we can use two species $R_i^A$ and $R_i^I$ as ``active'' and ``inactive'' species for register $r_i$, respectively.
The idea is that the total sum of the counts of species $R_i^A$ and $R_i^I$ is always equal to $b$:
whenever one is consumed, the other is produced.
Now, an $\textit{inc}(r_i,s_j)$ instruction could be implemented with the reaction $S_{j'} + R_i^I \rxn R_i^A + S_j$, and a $\textit{dec}(r_i,s_j, s_k)$ instruction could be implemented with the reactions $S_{j'} + R_i^A \rxn R_i^I + S_j$ and $S_{j'} + bR_i^I \rxn bR_i^I + S_k$. 
With this approach, register $r_i$ has a zero count exactly when inactive species $R_i^I$ has a count of $b$, and so we can zero-check without error.
Notice that this approach uses reactions with a large stoichiometric coefficient $b$.
At this point, there are two issues to be addressed: 
(1) how to generate an initial $b$ count of inactive species $R_i^I$, and
(2) how to transform the reactions into a series of third-order reactions
(avoiding the large stoichiometric coefficient $b$). 

Let's first consider a very simple construction which addresses the above concerns, albeit suboptimally.
Suppose $b = 2^n$ is a power of two. 
To handle (1), we can initially produce count $b$ of $R_i^I$ from a single copy of $A_1$ using $O(\log b)$ species with reactions 
\begin{align*}
A_1 &\rxn 2A_2 \\ 
A_2 &\rxn 2A_3 \\
&\vdots \\
A_{n} &\rxn R_i^I.    
\end{align*}
To handle (2), we can transform the decrement reactions into a series of $n$ bimolecular reactions by adding reversible versions of the reactions from (1) and ``counting down'' to some unique zero count indicator species $C_1$:
\begin{align*}
R_i^I &\revrxn C_{n} \\
2 C_{n} &\revrxn C_{n - 1} \\
&\vdots \\
2 C_2 &\revrxn C_1.
\end{align*}
Then $C_1$ is producible if and only if $R_i^I$ had count $\geq b$, so the reaction $S_{j} + C_1 \rxn S_{k} + C_1$ implements the ``jump to state $k$ if $r_i=0$'' portion of the $dec(r_i,s_j,s_k)$ command.
This construction allows an error-free simulation of a register machine with counters with bound $b$ exponential in the number of species.
Now we discuss a more sophisticated construction, based on previous results \cite{lipton1976reachability,cardoza1976exponential},
that achieves a counter bound $b$ that is \emph{doubly} exponential in the number of species.

\subsection{Counting to ${2^2}^n$ with $n$ species}
\label{sec:register-machine-count-doubly-exponential}
The CRN constructions in this paper simulate bounded register machines in the manner discussed previously.
Since we are focused on reducing the size of our CRN, we want to do this simulation with as few species (reactions) as possible.
Fortunately, we can rely on established results from prior work to do this.
Lipton provided a construction for which the largest producible amount of a species is a doubly exponential count~\cite{lipton1976reachability}.
However, this amount is only produced non-deterministically and (most) paths produce less.
Cardoza et al.\ went on to present a fully reversible system that can achieve this doubly exponential count as well \cite{cardoza1976exponential}.
Further, their system is halting in the sense that a new species is produced precisely when the maximum amount is reached.

While Cardoza et al.\ \cite{cardoza1976exponential} describe their construction in the language of commutative semigroup presentations, 
we present a modified construction in \Cref{fig:doubly-exponential-rxns} articulated as a CRN.
In the figure and in the text below, we use the ``box'' notation to indicate \emph{meta-reactions}, which correspond to a set of reactions.
Note that in \Cref{lem:metarxn} we will see that the combined behavior of the reactions in a meta-reaction module faithfully implement the meta-reaction semantics.
By construction, the sets of reactions that meta-reactions expand to overlap, and we include only one copy of any repeated reaction.
Each layer of the construction introduces $\Oh(1)$ more reactions and species---9 reactions ((1)--(9)) and 9 species ($S_i^k$, $H_i^k$, $X_i^k$, $T1_{i}^k, T2_{i}^k$, $C1_{i}^k$, $C2_{i}^k$, $C3_{i}^k$, $C4_{i}^k$) for each $i \in \{1,2,3,4\}$.

\begin{figure}
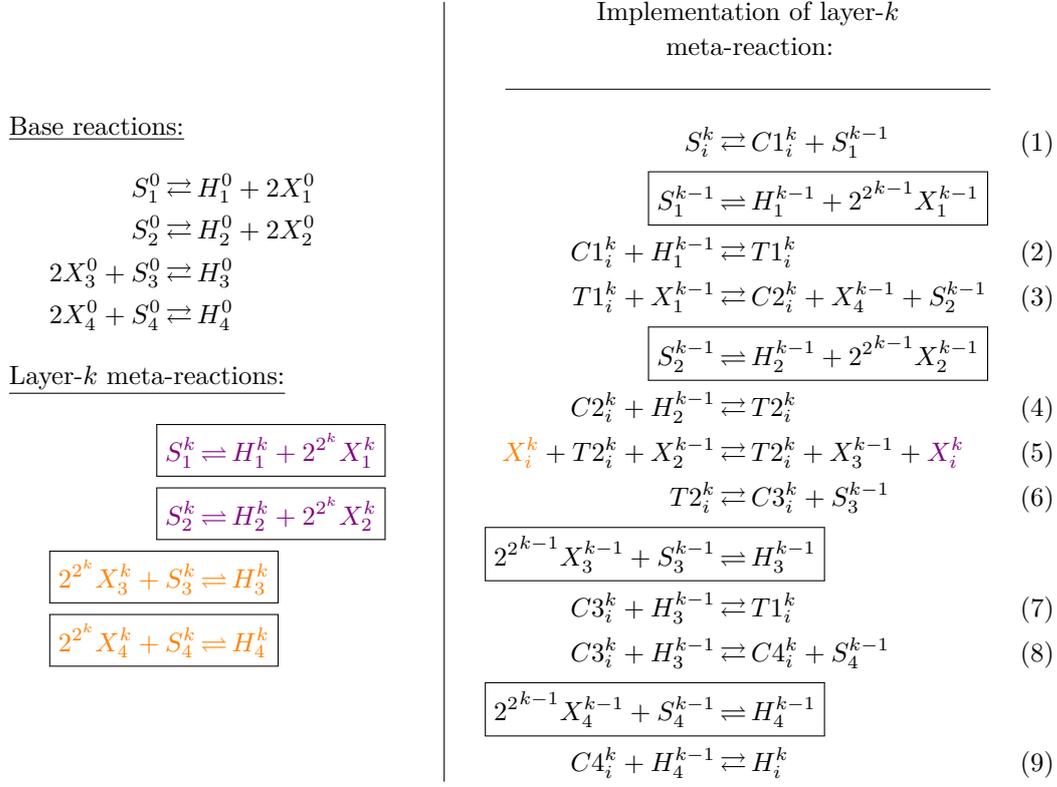

    \centering
    \begin{subfigure}[c]{0.4\textwidth}
        \underline{Base reactions:}
        \begin{align*}
            S_1^0 &\revrxn H_1^0 + 2 X_1^0\\
            S_2^0 &\revrxn H_2^0 + 2 X_2^0\\
            2 X_3^0 + S_3^0 &\revrxn H_3^0\\
            2 X_4^0 + S_4^0 &\revrxn H_4^0
        \end{align*}
        \underline{Layer-$k$ meta-reactions:}
        \begin{align*}
            \Aboxed{\color{violet} S_1^k &\metarevrxn H_1^k + 2^{2^k} X_1^k}\\
            \Aboxed{\color{violet} S_2^k &\metarevrxn H_2^k + 2^{2^k} X_2^k}\\
            \Aboxed{\color{orange} 2^{2^k} X_3^k + S_3^k &\metarevrxn H_3^k}\\
            \Aboxed{\color{orange} 2^{2^k} X_4^k + S_4^k &\metarevrxn H_4^k}
         \end{align*}
    \end{subfigure}
    \vrule
    \begin{subfigure}[c]{0.57\textwidth}
    \centering
    Implementation of layer-$k$ \\meta-reaction:
    \rule{0.8\textwidth}{0.1pt}
    \begin{align}
         S_i^k &\revrxn C1_{i}^k + S_1^{k-1} \tag{1}\\ 
         \Aboxed{S_1^{k-1} &\metarevrxn H_1^{k-1} + {2^2}^{k-1} X_1^{k-1}} \notag\\
         C1_{i}^k + H_1^{k-1} &\revrxn T1_{i}^k \tag{2}\\
         T1_{i}^k + X_1^{k-1} &\revrxn C2_{i}^k + X_4^{k-1} + S_2^{k-1} \tag{3}\\
         \Aboxed{S_2^{k-1} &\metarevrxn H_2^{k-1} + {2^2}^{k-1} X_2^{k-1}} \notag\\
         C2_{i}^k + H_2^{k-1} &\revrxn T2_{i}^k \tag{4}\\
         {\color{orange} X_i^{k}} + T2_{i}^k + X_2^{k-1} &\revrxn T2_{i}^k + X_3^{k-1} + {\color{violet} X_i^{k}} \tag{5}\\
         T2_{i}^k &\revrxn C3_{i}^k + S_3^{k-1} \tag{6}\\
         \Aboxed{{2^2}^{k-1} X_3^{k-1} + S_3^{k-1} &\metarevrxn H_3^{k-1}} \notag\\
         C3_{i}^k + H_3^{k-1} &\revrxn T1_{i}^k \tag{7}\\
         C3_{i}^k + H_3^{k-1} &\revrxn C4_{i}^k + S_4^{k-1} \tag{8}\\
         \Aboxed{{2^2}^{k-1} X_4^{k-1} + S_4^{k-1} &\metarevrxn H_4^{k-1}} \notag\\
         C4_{i}^k + H_4^{k-1} &\revrxn H_i^k \tag{9}
     \end{align}
     \end{subfigure}
     \caption{Doubly exponential counting construction. (Left) The base reactions and layer-$k$ meta-reactions. We use the box notation to indicate meta-reactions, which correspond to a set of reactions. Note that the reactions corresponding to the different meta-reactions overlap; when all the meta-reactions are expanded we include only one reaction copy. (Right) Explicit reactions for the layer-$k$ meta-reaction in terms of reactions and other meta-reactions. The core functionality is the same for any $i$, but reaction (5) either generates $X_i^k$'s (if $i \in \{1,2\}$) or consumes $X_i^k$'s (if $i \in \{3,4\}$).}
        \label{fig:doubly-exponential-rxns}
\end{figure}

The idea of the construction is to produce (or consume) a doubly exponential count of species $X$ by recursively producing (or consuming) quadratically more $X$'s than the previous layer.
Each species type performs a different role.
$X$ is the counting
species to be generated or consumed.
$S$ starts the process to generate/consume many molecules of species $X$.
$T$ transforms different types of $X$ species into one another.
$H$ indicates that the generation/consumption process has completed.
$C$ ``cleans up'' the $H$ species.
Reaction (5) in the meta-reaction implementation
(which converts $X_2^{k-1}$ into $X_3^{k-1}$) changes based on $i$. 
If $i \in \{1,2\}$, then $X_i$ appears as a product and is generated by this reaction.
If $i \in \{3,4\}$, the $X_i$ appears as a reactant and is consumed by this reaction.
A high level diagram of a layer-$k$ meta-reaction is shown in \Cref{fig:doubly-exponential-states},
which is helpful in understanding the behavior of the system.

\begin{figure}
    \centering
    \includegraphics[width=\textwidth]{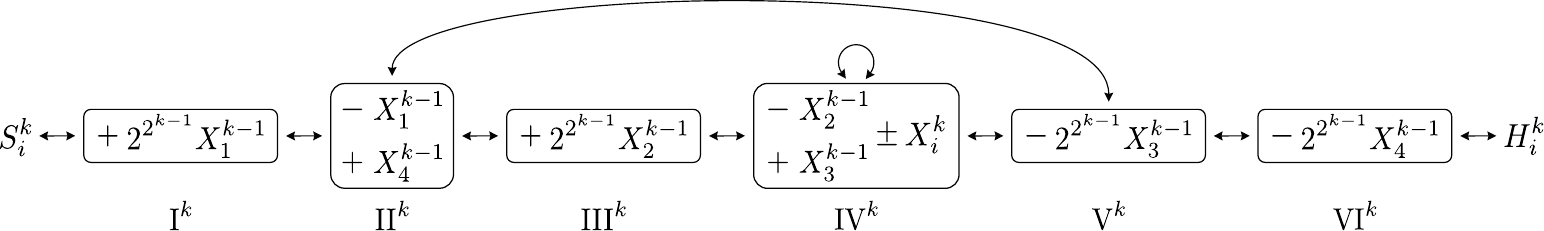}
    \caption{A visualization of the states for a layer-$k$ meta-reaction. The state transitions effectively execute a nested loop. In order to iterate the outer loop (transition from $\RN{5} \rightarrow \RN{2}$), the inner loop $\RN{4}$ must be executed $2^{2^{k-1}}$ times. And in order to leave state $\RN{6}$ and produce an $H$, the outer loop must be executed $2^{2^{k-1}}$. So, state $\RN{4}$ must be executed a total of $2^{2^{k}}$ times, which either produces or consumes that many $X_i^k$'s depending on the type of meta-reaction.}
    \label{fig:doubly-exponential-states}
\end{figure}

\begin{definition}\label{defn:well-led}
    Let $\vec{c}$ be a configuration of CRN $\crn{C}$ given above. 
    We say $\vec{c}$ is \emph{well-led} if $\vec{c}(S_*^*) + \vec{c}(H_*^*) + \vec{c}(T_*^*) = 1$ where the notation $S_*^*$ denotes any species with label $S$, regardless of the subscript or superscript. 
    In other words, there is only a single leader in the system and it either has the label $S$, $H$, or $T$.
    We call species $S_*^*$, $H_*^*$, and $T_*^*$ \emph{leader species}.
\end{definition}

\begin{observation}\label{obs:single-leader}
Every reaction has exactly one leader species as a reactant, and exactly one leader species as a product. 
\end{observation}

The following is immediate from \Cref{obs:single-leader}:

\begin{corollary}\label{corr:well-led-invariance}
    Let $\vec{c}$ be a well-led configuration of CRN $\crn{C}$ given above. 
    Then any configuration $\vec{d}$ such that $\vec{c} \reaches \vec{d}$ is also well-led. 
    In other words, the well-led property is forward invariant.
\end{corollary}

Informally, the observation above together with the well-led condition implies that we can reason about the meta-reactions in isolation,
without fear of cross-talk---because while one meta-reaction is executing, no reactions outside of it are applicable.
This allows us to inductively prove the main result of this section:

\begin{lemma}[Production]\label{lem:metarxn}
Consider the CRN implementing $\boxed{S_i^k \metarevrxn H_i^k + {2^2}^k X_i^k}$.
For any $n \in \N$, 
let $\vec{s} = \{n X_i^k, 1S_{i}^k\}$ and $\vec{h} = \{({2^2}^k + n) X_i^k, 1H_{i}^k\}$,
and let $\vec{c}$ be any configuration reachable from $\vec{s}$ or $\vec{h}$.
Then:
(a) Both $\vec{s}$ and $\vec{h}$ are reachable from $\vec{c}$.
(b) If $\vec{c}$ contains $S_{i}^k$ then $\vec{c} = \vec{s}$, and if $\vec{c}$ contains $H_{i}^k$ then $\vec{c} = \vec{h}$.
\end{lemma}
\begin{lemma}[Consumption]\label{lem:metarxn-consume}
Consider the CRN implementing $\boxed{{2^2}^k X_i^k + S_i^k \metarevrxn H_i^k}$.
For any $n \in \N$, 
let $\vec{s} = \{({2^2}^k + n)X_i^k, 1S_{i}^k\}$ and $\vec{h} = \{nX_i^k, 1H_{i}^k\}$,
and let $\vec{c}$ be any configuration reachable from $\vec{s}$ or $\vec{h}$.
Then:
(a) Both $\vec{s}$ and $\vec{h}$ are reachable from $\vec{c}$.
(b) If $\vec{c}$ contains $S_{i}^k$ then $\vec{c} = \vec{s}$, and if $\vec{c}$ contains $H_{i}^k$ then $\vec{c} = \vec{h}$.
\end{lemma}
\begin{proof}(Of \Cref{lem:metarxn} and \Cref{lem:metarxn-consume}, Sketch)
Both lemmas are proven by induction over the layers of the construction.
The base case ($k=1$) can be checked by inspection.
Now assume the lemmas are true for $k-1$ layers, and we want to prove them true for $k$ layers.

First we argue that the construction is correct if the $k-1$ layer meta-reactions are ``atomic'' and occur in one step.
As visualized in \Cref{fig:doubly-exponential-states}, 
the CRN
iterates through a nested loop process. 
Each state transition (states $\RN{1}^k$ through $\RN{6}^k$) is coupled to a conversion of the leader species;
the well-led condition ensures that the CRN is in exactly one state at any given time.
Each net forward traversal of the outer loop converts a $X_1^{k-1}$ to $X_4^{k-1}$,
and each forward traversal of the inner loop converts a $X_2^{k-1}$ to $X_3^{k-1}$.
Step $\RN{1}^k$ makes ${2^2}^{k-1} X_1^{k-1}$, bounding the net maximum number of times that the outer loop can happen in the forward direction.
Step $\RN{3}^k$ makes ${2^2}^{k-1} X_2^{k-1}$, bounding the net maximum number of times that the inner loop can happen in the forward direction for every net forward traversal of the outer loop.
This implies that reaction (5) can fire at most a net total ${2^2}^{k}$ times (producing at most a net total ${2^2}^{k} X_i^k$'s).

Step $\RN{5}^k$ consumes ${2^2}^{k-1} X_3^{k-1}$, requiring the net total number of forward traversals of the inner loop to be at least ${2^2}^{k-1}$ for every net forward traversal of the outer loop.
Step $\RN{6}^k$ consumes ${2^2}^{k-1} X_4^{k-1}$, requiring the net total number of forward traversals of the outer loop to be at least ${2^2}^{k-1}$.
This implies that reaction (5) must fire at least a net total ${2^2}^{k}$ times (producing at least a net total ${2^2}^{k} X_i^k$'s).

Thus, reaction (5) must be executed exactly ${2^2}^{k}$ times (producing exactly ${2^2}^{k} X_i^k$'s).
Notice that an excess of $X_i^k$ (as allowed by the statement of the lemma) does not affect the net total number of times reaction (5) can fire (forward or backward) since $X_2^{k-1}$ and $X_3^{k-1}$ are the limiting factors.

Now we need to make sure that this behavior is preserved once the meta-reactions are expanded to their constituent reactions.
Each meta-reaction $i$ in \Cref{fig:doubly-exponential-rxns} expands to some set $R_i$ of reactions.
First we note that for each meta-reaction, $R_i$ overlaps with reactions not in $R_i$ only over species $S_i^{k-1}$, $H_i^{k-1}$, and $X_i^{k-1}$.
We are not worried about cross-talk in species $S_i^{k-1}$ and $H_i^{k-1}$ because of the well-led property.
We may still be concerned, however, that external consumption of $X_i^{k-1}$ might somehow interfere with the meta-reaction.
Luckily, the well-led property and \Cref{obs:single-leader} enforce that unless we have $S_i^{k-1}$ or $H_i^{k-1}$ (i.e., we are at the beginning or end of the meta-reaction), it is never the case that a reaction in $R_i$ and a reaction not in $R_i$ are applicable at the same time.
Thus nothing outside the meta-reaction can change $X_i^{k-1}$ while the meta-reaction is executing.
\end{proof}

Note that although we chose to write \Cref{lem:metarxn} and \Cref{lem:metarxn-consume} separately, we could have just one kind of meta-reaction (production or consumption) and obtain the other kind by running the meta-reaction backward switching the roles of $S$ and $H$.
We include the two different versions because it is conceptually easier to just think about the intended execution being in the forward direction.

\section{Optimal Encoding}\label{sec:encoding}

\subsection{Encoding Information in CRNs}
\label{sec:simple-crn-producing-x}
In this section we discuss the encoding of an integer in a chemical reaction network.
In the same sense as Kolmogorov-optimal programs for Turing machines, we consider a similar measure of optimality for chemical reaction networks.
In particular, we ask the question, ``what is the smallest chemical reaction network that can produce a desired count of a particular chemical species?''

A simple construction shows that $x$ copies of some species can be produced using $\Oh(\log x)$ reactions.
The idea is to have a reaction for each bit $b_i$ of the binary expansion of $x$, and produce a copy of your output species in each reaction where $b_i = 1$.
More concretely, consider $\log x$ reactions of the form $X_i \rxn 2X_{i+1}$ and $X_i \rxn 2X_{i+1} + Y$.
For each bit $b_i$ in the binary expansion of $x$, use the first reaction if $b_i = 0$ and use the second reaction if $b_i = 1$.
Each species $X_i$ will have a count equal to $2^i$, and species $Y$ will have a count equal to the sum of the powers of two that were chosen (which is $x$).
While this simple construction generates $x$ with $\log x$ reactions, it is not immediately clear how to improve upon it.

Our first result shows how to construct a CRN that can generate $x$ copies of an output species (from an initial configuration with only a single molecule) yet uses only $\Oh(\log x/\log\log x)$ many reactions.
This matches the lower bound dictated by Kolmogorov complexity (see end of \Cref{sec:encoding}),
which suggests that the full power of CRNs is really being used in our construction.
Our construction is achieved through the simulation of (space-bounded) Turing machines via the simulation of (space-bounded) register machines.
A key aspect in this process is the ability of CRNs to use the previously discussed recursive counting technique to count very high with very few species (counting to $2^{2^k}$ with $k$ species).

\subsection{Our Construction}
Now, we present an encoding scheme to produce count $x$ of a particular species with $\Oh(n/\log n)$ CRN reactions, where $n = \log x$.
In the simple CRN given in \Cref{sec:simple-crn-producing-x},
each reaction encodes a single bit of $x$.
In the optimized construction with $k$ reactions,
each reaction will encode $\log k$ bits instead.\footnote{Adleman et al.\ \cite{adleman2001running} provided a clever base conversion trick for tile assembly programs. Here, we employ a permutation encoding trick to yield the same effect.}
A sketch of our construction is as follows:
\todo{[For journal version] DS: We using $k$ for a different thing now. But ok...}

\textit{Sketch: 
We start with a CRN in configuration $\vec{c}_1 = \{1 L\}$ and create a configuration $\vec{c}_2 = \{1 S_i, m_1R_1, m_2R_2, \dots, m_kR_k\}$ that represents a particular permutation of $k$ distinct elements.
We encode this permutation in the count of a species $I$, transforming configuration $\vec{c}_2$ into a configuration $\vec{c}_3 = \{1 S_j, m I\}$.
The count of species $I$ can be interpreted as the input to a Turing machine, so we simulate a Turing machine that maps the permutation to a unique integer via Lehmer code/factorial number system \cite{lehmer1960teaching, sedgewick1977permutation} (by choosing the right value of $k$, we can ensure there are sufficiently many permutations to let us map to $x$).
This Turing machine simulation transforms configuration $\vec{c}_3$ into configuration $\vec{c}_4 = \{1 H, x Y\}$.
}

\begin{theorem}\label{thm:incompressible}
For any $n \in \N$ and any $n$-bit integer $x$, there exists a chemical reaction network $\crn{C}_x$ that haltingly computes $x$ from initial configuration $\{1 L\}$ with $\size{\crn{C}_x} = \Oh(n/\log n)$.
\end{theorem}

\begin{proof}
    First, we describe how to construct CRN $\crn{C}_x$ that haltingly computes $x$ from starting configuration $\{1 L\}$, then we describe the size of $\size{\crn{C}_x}$.
    Let $k = \lceil n / \log n \rceil$.
    We will map a permutation of $k$ distinct elements to the integer $x$, and this value of $k$ ensures there are at least $x$ permutations.
    We break the construction into three primary steps.

    \textbf{Step 1: $\{1 L\} \reaches \{1 S_i,m_1R_1, m_2R_2,\dots, m_kR_k\}$.}
    We can transform $\{1 L\}$ into a configuration $\{1 S_i,m_1R_1, m_2R_2, \dots, m_kR_k\}$ where $(m_1, m_2, \dots, m_n)$ is a permutation of the integers 1 through $k$.
    This can be achieved with $k$ registers and $2k$ register machine states.
    For example, to set the permutation $(2,4,3,1)$, use instructions 
    \begin{align*}
        s_0: \ &\textit{inc}(r_2,s_1) \\
        s_1: \ &\textit{copy}(r_2, r_4, s_2)\\
        s_2: \ &\textit{inc}(r_2,s_3)\\
        s_3: \ &\textit{copy}(r_2, r_1, s_4)\\
        s_4: \ &\textit{inc}(r_2,s_5)\\
        s_5: \ &\textit{copy}(r_2, r_3, s_6)\\
        s_6: \ &\textit{inc}(r_2,s_7)
    \end{align*} 
    
    \textbf{Step 2: $\{1 S_i, m_1R_1, m_2R_2,\dots, m_kR_k\} \reaches \{1 S_j, m I\}$.}
    Now, we can transform configuration $\{1 S_i, m_1R_1, m_2R_2, \dots, m_kR_k\}$ into configuration $\{1 S_j, mI\}$, encoding the permutation as the integer count $m$ of species $I$. 
    For each register $r_i$ for $i$ from $1$ to $k$ in order, we can decrement the register to 0.
    On each decrement, we double the count of $I$ and then add 1 to it, i.e., appending a 1 to $m$'s binary expansion.
    After the register reaches 0, before moving to the next register, we double the count of $I$ again, appending a 0 to $m$'s binary expansion.
    For example, if the permutation configuration was $\{3 R_1, 1 R_2, 2 R_3\}$, the resulting count of $I$ in binary would be
    \[
    {\underbrace{111}_3 0 \underbrace{1}_1 0 \underbrace{11}_2}.
    \]
    
    \textbf{Step 3: $\{1 S_j, m I\} \reaches \{1 H, x Y\}$.}
    At this point, we can consider the value in register $I$, expressed as a binary string,
    to be the input tape content for a Turing machine that maps the permutation to the integer $x$ using a 
    standard Lehmer code/factorial number system technique\cite{lehmer1960teaching, sedgewick1977permutation}.
    The output of the Turing machine will be the count of $Y$ in configuration $\vec{c}$ at the end of the computation (with $\vec{c}(Y) = x$).
    Our register machine will have a state species that corresponds to the halted state of the Turing machine---and such a species serves as our halting species $H$.
    
    Now we argue the size of CRN $\size{\crn{C}_x} = \Oh(n/\log n)$, i.e., it uses $\Oh(k)$ reactions.
    The register machine program from Step 1 generates the permutation using $k$ registers and $2k$ register machine states, which results in $\Oh(k)$ CRN reactions.
    The register machine program from Step 2 encodes the permutation as a binary number in register $I$ using $\Oh(k)$ registers and $\Oh(k)$ register machine states, which also results in $\Oh(k)$ CRN reactions.
    Even a naive algorithm for the Turing machine from Step 3 maps the permutation to an integer using $\Oh(k^2\log k)$ space ($\Oh(k^2)$ bits to store the initial permutation, $\Oh(k\log k)$ bits to store the Lehmer code, $\Oh(k^2\log k)$ bits to store factorial bases $1!$ through $k!$, and $\Oh(k\log k)$ bits to store the integer $x$).
    Recall, a Turing machine using space $\Oh(k^2\log k)$ can be simulated by a register machine with count bound $\Oh(2^{k^2\log k})$ on its registers.
    This can in turn be simulated by a CRN via the construction of \Cref{sec:register-machine-count-doubly-exponential} with $\Oh( \log \log 2^{k^2\log k} ) = \Oh(\log k)$ reactions. 
    Thus $\Oh(k)$ reactions suffices to simulate the register machine instructions as well as the bounded counters for our register machine to simulate this Turing machine.
\end{proof}

The above construction is optimal for almost all integers $x$ in the following sense.
Any CRN of $\size{\crn{C}}$ reactions, each with $\Oh(1)$ reactants and products, can be encoded in a string of length $\Oh(\size{\crn{C}} \log{\size{\crn{C}}})$.
Given an encoded CRN stably computing an integer $x$, 
a fixed-size program can simulate it and return $x$.
Thus $\K(x) \leq \Oh(\size{\crn{C}} \log{\size{\crn{C}}})$.
The pigeonhole principle argument for Kolmogorov complexity implies that  $\K(x) < \lceil \log x \rceil - \Delta$ for at most $(1/2)^{\Delta}$ of all $x$ \cite{li2008introduction}.
Together these observations imply that there is a $c$ such that for most $x$ there does not exist a CRN $\crn{C}$ of size smaller than $c n/\log n$ that stably computes $x$.

\section{Algorithmic Compression}\label{sec:compression}

The construction in Section~\ref{sec:encoding} is optimal for incompressible integers (integers $x$ where $\mathrm{K}(x) \approx \size{x}$, which is the case for ``most'' integers).
Now we extend the construction to be optimal within global multiplicative constants for \emph{all} integers. 
For algorithmically compressible integers $x$, there exists a $p$ such that $\U(p) = x$ and $\size{p} < \size{x}$.
We discuss the construction in Section~\ref{subsec:compressible-construction} and we argue optimality of our construction in Section~\ref{subsec:optimality}.

\subsection{Our Construction}\label{subsec:compressible-construction}
We now show how to fully exploit the encoding scheme and doubly exponential counter from Section~\ref{sec:encoding} to achieve an optimal result for all integers.
A sketch of our construction is as follows:

\emph{Sketch: Given a program $p$ for a fixed Universal Turing Machine $\U$ such that $\U(p) = x$, we construct a CRN that simulates running $p$ on $\U$ via a register machine simulation. 
The idea is to use $\Oh(\size{p} / \log \size{p})$ reactions to encode $p$, and to use $\Oh(\log(\texttt{space}(\U(p)))$ reactions for a counter machine simulation of $\U(p)$.}

\begin{theorem}\label{thm:compressible}
For any integer $x$, there exists a CRN $\crn{C}_x$ that haltingly computes $x$ from initial configuration $\{1 L\}$ with 
$\size{\crn{C}_x} = \Oh(\Kc(x))$.
\end{theorem}

\begin{proof}
    Let $p$ be a program for a fixed Universal Turing Machine $\U$ such that $\U(p)=x$.
    We encode $p$ in the manner provided by \Cref{thm:incompressible} using $\Oh(\size{p}/\log\size{p})$ reactions.
    This results in $p$ count of species $Y$ (specifically, configuration $\{1 H, p Y\}$).
    Since haltingly-computing CRNs are composable via concatenation \cite{chalk2019composable, severson2019composable}, we can consider $\{1 H, p Y\}$ to be taken as the input for another system which simulates running $\U(p)$ via the previously described register machine method with bounded register count (\Cref{sec:register-machine-count-doubly-exponential}).
    Again, we need enough species/reactions to ensure our bounded registers can count high enough.
    The registers must be able to store an integer that represents the current configuration of the Turing machine being simulated (at most this is $2^{\texttt{space}(\U(p))}$).
    Since we have doubly exponential counters, an additional $\log{\big(}\texttt{space}(\U(p)){\big)}$ species are needed to do this.
    So, the total size of our CRN is $\Oh\big{(}\size{p}/\log\size{p} + \log(\texttt{space}(\U(p)))\big{)}$ and by choosing the program $p$ that minimizes this expression, we see $\size{\crn{C}_x} = \Oh(\Kc(x))$.
\end{proof}

It is interesting to note the appearance of our ``space-aware'' version of Kolmogorov complexity.
Importantly, this notion is different from space-bounded Kolmogorov complexity that puts a limit on the space usage of the program that outputs $x$.
This alternate version allows a trade-off between compact program descriptions and the space required to run those programs, which seems natural for systems like CRNs.
Perhaps it is surprising
that this (computable) measure of complexity shows up here, and at first it may seem like $\log$ of this space usage is a bit arbitrary, but we will show that this is indeed optimal (within global multiplicative constants) for CRNs.

\subsection{Optimality}\label{subsec:optimality}
Here, we argue that size of CRN $\crn{C}_x$ from \Cref{thm:compressible} is optimal.
We begin by giving a definition for the size of the optimal CRN that haltingly computes an integer $x$.

\begin{definition}
For any integer $x$, define $\Kcrn(x) = \min\{\size{\crn{C}} : \text{CRN } \crn{C} \text{ haltingly computes } x\}$. In other words, $\Kcrn(x)$ is the size of the smallest CRN that haltingly computes $x$.
\end{definition}

Our argument relies on a Turing machine that solves the coverability problem for CRNs.
We give the definition for this problem in \Cref{def:coverability} and discuss its space complexity in \Cref{lem:TM_space}.

\begin{definition}[Coverability]\label{def:coverability}
    Given a CRN $\crn{C}$, initial configuration $\vec{s}$, and target configuration $\vec{u}$, does there exist a configuration $\vec{t} \geq \vec{u}$ such that $\vec{s} \reaches \vec{t}$?
\end{definition}

Using the natural notion of problem size $n$ for the specification of a coverability problem, 
Lipton provided a $2^{\Omega(\sqrt{n})}$ space lower bound for coverability \cite{lipton1976reachability}, which was later improved to $2^{\Omega(n)}$ by Mayr and Meyer \cite{mayr1982complexity}.
As for upper bounds, Rackoff provided an algorithm to decide coverability that uses $2^{\Oh(n\log n)}$ space \cite{rackoff1978covering}.
Following this, Koppenhagen and Mayr gave an algorithm that decided coverability in $2^{\Oh(n)}$ space for \emph{reversible} systems, closing the gap for this class of systems \cite{koppenhagen2000optimal}.
A recent result by K\"{u}nnemann et al.\ also closes this gap \cite{kunnemann2023coverability} for Vector Addition Systems with States.
Our work uses this latest result.

\begin{lemma}[Implied by Theorem 3.3 from \cite{kunnemann2023coverability}]\label{lem:TM_space}
    Let CRN $\crn{C} = (\species, \reactions)$ be a CRN that haltingly computes $x$.
    Then there exists an algorithm which solves coverability for CRN $\crn{C}$ for initial configuration $\{1 L\}$ and target configuration $\{1 H\}$ which uses $2^{\Oh(\size{\crn{C}})}$ space.
\end{lemma}

\begin{proof}
    This result follows from Theorem 3.3 from the recent work by K\"{u}nnemann et al.~\cite{kunnemann2023coverability}.
    There, the authors consider the problem of coverability in Vector Addition Systems with States (VASS).
    They show that if the answer to coverability is yes, the length of the longest path is $n^{2^{\Oh(d)}}$ where $n$ is the maximum value change of any transition and $d$ is the dimension of the vector.
    For us, $n = 2$ since each reaction has at most two reactants/products and $d = \size{\species}$.
    While vector addition systems are not capable of ``catalytic'' transitions, it is known that the same effect can be achieved by decomposing transitions into two vector additions.
    So the path length would at most double for our systems.

    With this bound on the path length, we can consider an algorithm that non-deterministically explores the state space of $\crn{C}$ (from starting configuration $\{1 L\}$) by simulating reactions on a current configuration of the system until a configuration that covers $\{1 H\}$ is found. 
    By Savitch's Theorem~\cite{savitch1970relationships}, this can be converted to a deterministic algorithm using the same space: $\Oh(\size{\reactions}\log\size{\reactions})$ bits to hold a description of $\crn{C}$, $2^{\Oh(\size{\species})}$ bits for a path length counter, and $2^{\Oh(\size{\species})}\cdot\log(\size{S})$ bits to store the current configuration of $\crn{C}$. 
    All of these values are absorbed under a $2^{\Oh(\size{\crn{C}})}$ bound.
\end{proof}

With this space bound on the coverability problem established, we can now argue that the size of our constructed CRN from \Cref{thm:compressible} is asymptotically equal to the size of the the smallest CRN that haltingly computes $x$.

\begin{lemma}\label{thm:optimality}
    For all $x \in \mathbb{N}$,
    letting $\crn{C}_x$ be the CRN from \Cref{thm:compressible},
    $\size{\crn{C}_x} = \Theta(\Kcrn(x)).$
\end{lemma}

\begin{proof}
    Clearly $\size{\crn{C}_x} = \Omega(\Kcrn(x))$, by definition of $\Kcrn(x)$ and since $\crn{C}_x$ from \Cref{sec:compression} is an instance of a CRN that haltingly computes $x$.
    Now, we argue that $\size{\crn{C}_x} = \Oh(\Kcrn(x))$.
    The big picture is that one of the programs over which $\Kc(x)$ is minimized in the construction of $\crn{C}_x$ is the program solving coverability for the optimal CRN for generating $x$.
    
    Start with the CRN $\crn{K} = (\species_{\crn{K}}, \reactions_{\crn{K}})$ that haltingly computes $x$ with optimal size $\size{\crn{K}} = \Kcrn(x) = n$. 
    Consider program $p_{\crn{K}}$ that solves coverability for $\crn{K}$ with initial configuration $\{1 L\}$ and target configuration $\{1 H\}$, and outputs $x$. 
    Now, with that program $p_{\crn{K}}$, build a CRN $\crn{K}'$ by following our construction for \Cref{thm:compressible}. 
    We know $\size{p_{\crn{K}}} = \Oh(n\log n)$ so our final CRN needs $\Oh(n)$ reactions to encode $p_{\crn{K}}$ (by \Cref{thm:incompressible}).
    By \Cref{lem:TM_space}, we know that the space usage of $\U(p_{\crn{K}})$ is $2^{\Oh(\size{\crn{K}})}$,
    so our final CRN needs $\Oh(\size{\crn{K}})$ additional reactions to have large enough registers for the simulation of $\U(p_{\crn{K}})$. 
    Thus, the total size of our final CRN is $\Oh(\size{\crn{K}} + \size{\crn{K}}) = \Oh(\Kcrn(x))$.
\end{proof}

The following theorem, which is the main result of our paper, follows immediately from \Cref{thm:compressible} and \Cref{thm:optimality}.
It characterizes the optimal number of reactions haltingly computing a number $x$ using the space-aware Kolmogorov complexity measure $\Kc$ defined in \Cref{sec:kolmogorov-complexity}.
\begin{theorem}
    \label{thm:main}
    For all $x \in \mathbb{N}$,
    $\Kcrn(x) = \Theta(\Kc(x))$.
\end{theorem}
Although, as mentioned above, CRN stable computation is not Turing universal, the theorem underlines its essential connection to space-bounded Turing machine computation.

\section{Open Questions}\label{sec:conclusion}

Our results rely on the fact that we consider CRNs that perform halting computation---the end of the computation is indicated by the production of a designated halting species.
This constraint, intuitively that the systems know when they have finished a computation, is rather strong.
It is known that a much larger class of functions can be stably computed than can be haltingly computed~\cite{p1ccrnJournal}.
It remains an open question if lifting this halting requirement (and allowing just stable computation) reduces the reaction complexity.

It is also worth noting that our approach starts with exactly one copy of a special leader species.
Recently, Czerner showed that leaderless protocols are capable of deciding doubly exponential thresholds \cite{czerner2022leaderless}.
While starting in some uniform state and converging to a specific state would be a better expression of ``chemical self-organization,'' their construction 
seems incompatible with our register machine simulation.
Leaderless stable integer computation remains an area for future work.

Making a tight connection between stable integer computation and counting predicate computation commonly studied in population protocols \cite{czerner2022leaderless} also remains open.
We can easily follow the halting generation of a specific amount of $x$ by running the ``less-than-or-equal-to'' predicate, thereby converting our constructions to compute a counting predicate with only a constant more reactions. 
This gives a new general upper bound on the complexity of counting in terms of $\Kc(x)$.
However, it is unclear whether counting predicate constructions carry over to the generation problem, leaving it open whether counting may be easier.

Our notion of ``space-aware'' Kolmogorov complexity $\Kc$ is interesting in its own right.
While the similar quantity $\Ks$ has been previously studied in the context of computational complexity theory~\cite{allender2011pervasive} (see also \Cref{sec:kolmogorov-complexity}),
it is not clear which properties proven of $\Ks$ carry over to $\Kc$. 
Although the robustness to the choice of $\U$ carries over,
other properties may not.
For example, it is not obvious whether our results still hold if we consider programs that output a single bit of $x$ at a time (like $\Ks$ does). 

A core piece of this work is simulating space-bounded Turing machines, so it is very natural to extend the discussion to Boolean circuits (computing functions $\phi:\{0,1\}^n \to \{0,1\}$).
When attempting to compute Boolean functions with CRNs, one may be tempted to directly implement a Boolean circuit by creating $\Oh(1)$ reactions per gate in the circuit.
However, our results imply that reaction complexity can be improved by doing a space-bounded Turing machine simulation instead---when the circuit is algorithmically ``compressible.''
An important class of such compressible circuits are uniform circuits, i.e., those constructable by a fixed Turing machine given an input size.
Prior work established a quadratically tight connection between the depth of uniform circuits and Turing machine space~\cite{borodin1977relating}.
Further investigation into optimal Boolean function computation is warranted.

\bibliography{refs}

\end{document}